\documentclass[journal]{IEEEtran}

\usepackage{graphicx}
\newtheorem{theorem}{\bf Theorem}

\begin{document}
%
\title{Approaching Blokh-Zyablov Error Exponent with Linear-Time Encodable/Decodable Codes }
%
%
\author{Zheng~Wang, ~\IEEEmembership{Student Member,~IEEE}, Jie~Luo, ~\IEEEmembership{Member,~IEEE}
\thanks{The authors are with the Electrical and Computer Engineering Department, Colorado State University, Fort Collins, CO 80523. E-mail: \{zhwang, rockey\}@engr.colostate.edu. }
}

\maketitle

\begin{abstract}
Guruswami and Indyk showed in \cite{ref Guruswami05} that Forney's error exponent can be achieved with linear coding complexity over binary symmetric channels. This paper extends this conclusion to general discrete-time memoryless channels and shows that Forney's and Blokh-Zyablov error exponents can be arbitrarily approached by one-level and multi-level concatenated codes with linear encoding/decoding complexity. The key result is a revision to Forney's general minimum distance decoding algorithm, which enables a low complexity integration of Guruswami-Indyk's outer codes into the concatenated coding schemes.
\end{abstract}

\begin{keywords}
coding complexity, concatenated code, error exponent
\end{keywords}

%
\IEEEpeerreviewmaketitle


\section{Introduction}
\label{SectionI}
Consider communication over a discrete-time memoryless channel modeled by a conditional point mass function (PMF) or probability density function (PDF) $p_{Y|X}(y|x)$, where $x\in X$ and $y\in Y$ are the input and output symbols, $X$ and $Y$ are the input and output alphabets, respectively. Let $\mathcal{C}$ be the Shannon capacity. Fano showed in \cite{ref Fano61} that the minimum error probability $P_e$ for block channel codes of rate $R$ and length $N$ is bounded by
\begin{equation}
\lim_{N\to \infty}-\frac {\log P_e}{N}\ge E(R),
\end{equation}
where $E(R)$ is a positive function of channel transition probabilities, known as the error exponent. For finite input and output alphabets, without coding complexity constraint,  the maximum achievable $E(R)$ is given by Gallager in \cite{ref Gallager65},
\begin{equation}\label{GallagerE}
E(R)= \max_{p_X}E_L(R,p_X),
\end{equation}
where $p_X$ is the input distribution, and $E_L(R,p_X)$ is given for different values of $R$ as follows,
\begin{eqnarray}\label{GallagerE1}
\begin{array}{ll} \max_{\rho \ge1}\left \{-\rho R+E_x(\rho, p_X)\right \}& 0 \le R \le R_x \\
-R+E_{0}(1,p_X) &  R_x \le R \le R_{crit} \\
\max_{0\le\rho\le1}\left \{-\rho R+E_0(\rho, p_X)\right \} & R_{crit} \le R \le \mathcal{C}.
\end{array}
\end{eqnarray}
The definitions of other variables in (\ref{GallagerE1}) can be found in \cite{ref Forney66}. If we replace the PMF by PDF, the summations by integrals and the $\max$ operators by $\sup$ in (\ref{GallagerE}), (\ref{GallagerE1}), the maximum achievable error exponent for continuous channels, i.e., channels whose input and/or output alphabets are the set of real numbers \cite{ref Gallager65}, is still given by (\ref{GallagerE}).

In \cite{ref Forney66}, Forney proposed a one-level concatenated coding scheme, which can achieve the following error exponent, known as Forney's exponent, for any rate $R<\mathcal{C}$ with a complexity of $O(N^4)$.
\begin{equation}
E_c(R)=\max_{r_o\in\left [\frac{R}{\mathcal{C}},1 \right ]}(1-r_o)E\left (\frac{R}{r_o}\right ),
\label{EcR}
\end{equation}
where $r_o$ and $R$ are the outer and the overall rates, respectively. Forney's coding scheme concatenates a maximum distance separable (MDS) outer error-correction code with well performed inner channel codes. To achieve $E_c(R)$, the decoder is required to exploit reliability information from the inner codes using a general minimum distance (GMD) decoding algorithm \cite{ref Forney66}. Forney's GMD algorithm essentially carries out outer code decoding, under various conditions, for $O(N)$ times. The overall decoding complexity of $O(N^4)$ is due to the fact that the outer code (which is a Reed-Solomon code) used in \cite{ref Forney66} has a decoding complexity of $O(N^3)$. Forney's concatenated codes were generalized to multi-level concatenated codes, also known as the generalized concatenated codes, by Blokh and Zyablov in \cite{ref Blokh82}. As the order of concatenation goes to infinity, the error exponent approaches the following Blokh-Zyablov bound (or Blokh-Zyablov error exponent) \cite{ref Blokh82}\cite{ref Barg05}.
\begin{equation}\label{BZBound}
E^{(\infty)}(R)=\max_{p_X, r_o\in {\left[\frac{R}{\mathcal{C}},1\right ]}}\left (\frac{R}{r_o}-R\right)\left [\int^{\frac{R}{r_o}}_0\frac{dx}{E_L(x, p_X)}\right ]^{-1}.
\end{equation}

In \cite{ref Guruswami05}, Guruswami and Indyk proposed a family of linear-time encodable/decodable nearly MDS error-correction codes. By concatenating these codes (as outer codes) with {\it fixed-lengthed} binary inner codes, together with Justesen's GMD algorithm \cite{ref Justesen72}, Forney's error exponent was shown to be achievable over binary symmetric channels (BSCs) with a complexity of $O(N)$ \cite{ref Guruswami05}, i.e., linear in the codeword length. The number of outer code decodings required by Justesen's GMD algorithm is only a constant\footnote{Strictly speaking, the required number of outer code decodings is linear in the inner codeword length, which is fixed at a reasonably large constant.}, as opposed to $O(N)$ in Forney's case \cite{ref Forney66}. Since each outer code decoding has a complexity of $O(N)$, upper-bounding the number of outer code decodings by a constant is required for achieving the overall linear complexity. Because Justesen's GMD algorithm assumes binary channel outputs \cite{ref Justesen72}\cite{ref Guruswami01}, achievability of Forney's exponent was only proven for BSCs in \cite[Theorem 8]{ref Guruswami05}.

In this paper, we show that Forney's GMD algorithm can be revised to carry out outer code decoding for only a constant number of times\footnote{The revision can also be regarded as an extension to Justesen's GMD decoding given in \cite{ref Justesen72}.}. With the help of the revised GMD algorithm, by using Guruswami-Indyk's outer codes with fixed-lengthed inner codes, one-level and multi-level concatenated codes can arbitrarily approach Forney's and Blokh-Zyablov exponents with linear complexity, over general discrete-time memoryless channels.

\section{Revised GMD Algorithm and Its Impact on Concatenated Codes}
\label{SectionII}

Consider one-level concatenated coding schemes. Assume, for an arbitrarily small $\varepsilon_1>0$, we can construct a linear encodable/decodable outer error-correction code, with rate $r_o$ and length $N_o$, which can correct $t$ symbol errors and $d$ symbol erasures so long as $2t+d<N_o(1-r_o-\varepsilon_1)$. Note that this is possible for large $N_o$ as shown by Guruswami and Indyk in \cite{ref Guruswami05}. To simplify the notations, we assume $N_o(1-r_o-\varepsilon_1)$ is an integer. The outer code is concatenated with suitable inner codes with rate $R_i$ and fixed length $N_i$. The rate and length of the concatenated code are $R=r_oR_i$ and $N=N_oN_i$, respectively. In Forney's GMD decoding, inner codes forward not only the estimates $\hat{\mbox{\boldmath $x$}}_m=[\hat{x}_1, \dots, \hat{x}_i, \dots, \hat{x}_{N_o}]$ but also a reliability vector $\mbox{\boldmath $\alpha$}=[\alpha_1, \dots, \alpha_i, \dots, \alpha_{N_o}]$ to the outer code, where $\hat{x}_i\in GF(q)$, $0\le \alpha_i \le 1$ and $1 \le i \le N_o$. Let
\begin{equation}
s(\hat{x},x)= \left \{ \begin{array}{ll}+1 & x=\hat{x}\\
-1 & x\neq \hat{x}\end{array} \right . .
\end{equation}
For any outer codeword $\mbox{\boldmath $x$}_m=[x_{m1},x_{m2},\dots,x_{mN_o}]$, define a dot product $\mbox{\boldmath $\alpha$}\cdot\mbox{\boldmath $x$}_m$ as follows
\begin{equation}
\mbox{\boldmath $\alpha$}\cdot\mbox{\boldmath $x$}_m = \sum_{i=1}^{N_o} \alpha_is(\hat{x}_i, x_{mi})=\sum_{i=1}^{N_o} \alpha_is_i.
\end{equation}

\begin{theorem}{\label{Theorem1}}
There is at most one codeword $\mbox{\boldmath $x$}_m$ that satisfies
\begin{equation}
\mbox{\boldmath $\alpha$}\cdot\mbox{\boldmath $x$}_m>N_o(r_o+\varepsilon_1).
\end{equation}
\end{theorem}

Theorem \ref{Theorem1} is implied by Theorem 3.1 in \cite{ref Forney66}.

Rearrange the weights in ascending order of their values and let $i_1, \dots, i_j, \dots, i_{N_o}$ be the indices such that
\begin{equation}
\alpha_{i_1}\le \dots\le \alpha_{i_j}\le \dots\le \alpha_{i_{N_o}}.
\end{equation}
Define $\mbox{\boldmath$q$}_k=[q_k(\alpha_1),\dots, q_k(\alpha_j), \dots, q_k(\alpha_{N_o})]$, for $0 \le k < 1/\varepsilon_2$, where $\varepsilon_2>0$ is a positive constant with $1/\varepsilon_2$ being an integer, and $q_k(\alpha_{i_j})$ is given by
\begin{eqnarray}
q_k(\alpha_{i_j})= \left \{ \begin{array}{ll}0 & \mbox{if} \quad \alpha_{i_j} \le k\varepsilon_2\quad \\ & \mbox{and}\quad i_j\le N_o(1-r_o-\varepsilon_1)\\
1 & \mbox{otherwise}\end{array} \right . .
\end{eqnarray}
Define dot product $\mbox{\boldmath $q$}_k\cdot \mbox{\boldmath $x$}_m$ as
\begin{equation}
\mbox{\boldmath $q$}_k\cdot \mbox{\boldmath $x$}_m=\sum_{i=1}^{N_o} q_k(\alpha_{i})s(\hat{x}_i,x_{mi})=\sum_{i=1}^{N_o} q_k(\alpha_i)s_i.
\end{equation}
Then following theorem gives the key result that enables the revision of Forney's GMD decoder.

\begin{theorem}{\label{Theorem2}}
If $\mbox{\boldmath$\alpha$}\cdot\mbox{\boldmath $x$}_m>N_o\left (\frac{\varepsilon_2}{2}+(r_o+\varepsilon_1)(1-\frac{\varepsilon_2}{2})\right)$, then for some
$0 \le k < 1/\varepsilon_2$, $\mbox{\boldmath$q$}_k\cdot$$\mbox{\boldmath $x$}_m>N_o(r_o+\varepsilon_1)$.
\end{theorem}

\begin{proof}
Define a set of values $c_j=(j-1/2)\varepsilon_2$ for $1 \le j \le 1/\varepsilon_2$ and an integer $p=\lceil \alpha_{i_{N_o(1-r_o-\varepsilon_1)}}/\varepsilon_2\rceil$, where $1\le p\le 1/\varepsilon_2$. \footnote{Note that the value of $p$ cannot be $0$. Because if $p=0$, i.e., $\alpha_{i_{N_o(1-r_o-\varepsilon_1)}}=0$, then there are at least $N_o(1-r_o-\varepsilon_1)$ zeros in vector $\mbox{\boldmath $\alpha$}$. Consequently, $\mbox{\boldmath $\alpha$}\cdot\mbox{\boldmath $x$}_m\le N_o(r_o+\varepsilon_1)<N_o\left(\frac{\varepsilon_2}{2}+(r_o+\varepsilon_1)\left(1-\frac{\varepsilon_2}{2}\right)\right)$, which contradicts the assumption that $\mbox{\boldmath$\alpha$}\cdot\mbox{\boldmath $x$}_m>N_o\left (\frac{\varepsilon_2}{2}+(r_o+\varepsilon_1)(1-\frac{\varepsilon_2}{2})\right)$.}

Let
\begin{eqnarray}
&& \lambda_0 = c_1 \nonumber\\
&& \lambda_k = c_{k+1}-c_k, 1\le k\le p-1 \nonumber\\
&& \lambda_{p} = \alpha_{i_{N_o(1-r_o-\varepsilon_1)+1}}-c_p \nonumber\\
&& \lambda_h = \alpha_{i_{h-p+N_o(1-r_o-\varepsilon_1)+1}}-\alpha_{i_{h-p+N_o(1-r_o-\varepsilon_1)}},\nonumber \\
&& \qquad \mbox{if}\quad p< h<p+N_o(r_o+\varepsilon_1)\nonumber\\
&& \lambda_{p+N_o(r_o+\varepsilon_1)}= 1-\alpha_{i_{N_o}}.
\end{eqnarray}
We have
\begin{equation}
\sum_{k=0}^{j-1}\lambda_k= \left\{ \begin{array}{ll}c_j & 1 \le j \le p\\ \alpha_{i_{j-p+N_o(1-r_o-\varepsilon_1)}} & p < j \le p+N_o(r_o+\varepsilon_1)\end{array} \right . ,
\end{equation}
and
\begin{equation}
\sum_{k=0}^{p+N_o(r_o+\varepsilon_1)}\lambda_k=1.
\end{equation}

Define a new weight vector $\tilde{\mbox{\boldmath $\alpha$}}=[\tilde{\alpha}_1, \dots, \tilde{\alpha}_i, \dots, \tilde{\alpha}_{N_o}]$ with
\begin{equation}
\tilde{\alpha}_i= \left\{ \begin{array}{ll} \mbox{argmin}_{c_j, 1\le j \le p}|c_j-\alpha_i| & \alpha_i\le \alpha_{i_{N_o(1-r_o-\varepsilon_1)}}\\ \alpha_i & \alpha_i>\alpha_{i_{N_o(1-r_o-\varepsilon_1)}} \end{array} \right . .
\end{equation}
Define $\mbox{\boldmath $p$}_k=[p_k(\alpha_1), \dots, p_k(\alpha_i), \dots, p_k(\alpha_{N_o}) ]$ with $1\le k\le p+N_o(r_o+\varepsilon_1)$ such that for $0\le k<p$
\begin{equation}
\mbox{\boldmath $p$}_k = \mbox{\boldmath $q$}_k,
\end{equation}
and for $p\le k\le p+N_o(r_o+\varepsilon_1)$
\begin{equation}
\mbox{\boldmath $p$}_k(\alpha_i) = \left\{ \begin{array}{ll}0 & \alpha_i\le \alpha_{i_{k-p+N_o(1-r_o-\varepsilon_1)}}\\ 1 & \alpha_i>\alpha_{i_{k-p+N_o(1-r_o-\varepsilon_1)}} \end{array} \right . .
\end{equation}
We have
\begin{equation}
\mbox{\boldmath $\tilde{\alpha}$}=\sum_{k=0}^{p+N_o(r_o+\varepsilon_1)}\lambda_k\mbox{\boldmath $p$}_k.
\end{equation}

Define a set of indices
\begin{equation}
\mathcal{U}=\{i_1, i_2, \ldots, i_{N_o(1-r_o-\varepsilon_1)}\}.
\end{equation}
According to the definition of $\tilde{\alpha}_i$, for $i \notin \mathcal{U}$, $\tilde{\alpha}_i=\alpha_i$. Hence
\begin{equation}
\tilde{\mbox{\boldmath $\alpha$}} \cdot \mbox{\boldmath $x$}_m=\mbox{\boldmath $\alpha$} \cdot \mbox{\boldmath $x$}_m+\sum_{i \in \mathcal{U}}\left (\tilde{\alpha}_i-\alpha_i\right )s_i.
\end{equation}
Since $|\tilde {\alpha}_i-\alpha_i| \le \varepsilon_2/2$, and $s_i=\pm 1$, we have
\begin{equation}
\sum_{i \in \mathcal{U}}\left (\tilde {\alpha}_i-\alpha_i\right )s_i \ge -N_o(1-r_o-\varepsilon_1)\frac{\varepsilon_2}{2}.
\end{equation}
Consequently, $\mbox{\boldmath $\alpha$} \cdot \mbox{\boldmath $x$}_m>N_o\left(\frac{\varepsilon_2}{2}+(r_o+\varepsilon_1)\left(1-\frac{\varepsilon_2}{2}\right)\right)$ implies
\begin{equation}\label{AlphaX}
\mbox{\boldmath $\tilde{\alpha}$} \cdot \mbox{\boldmath $x$}_m>N_o(r_o+\varepsilon_1).
\end{equation}

If $\mbox{\boldmath $p$}_k \cdot \mbox{\boldmath $x$}_m\le N_o(r_o+\varepsilon_1)$ for all $\mbox{\boldmath $p$}_k$'s, then
\begin{eqnarray}
\mbox{\boldmath $\tilde{\alpha}$} \cdot \mbox{\boldmath $x$}_m&=&\sum_{k=0}^{p+N_o(r_o+\varepsilon_1)}\lambda_k\mbox{\boldmath $p$}_k\cdot \mbox{\boldmath $x$}_m\nonumber\\
&\le& N_o(r_o+\varepsilon_1)\sum_{k=0}^{p+N_o(r_o+\varepsilon_1)}\lambda_k \nonumber\\
&=&N_o(r_o+\varepsilon_1),
\end{eqnarray}
which contradicts (\ref{AlphaX}). Therefore, there must be some $\mbox{\boldmath $p$}_k$ that satisfies
\begin{equation}\label{PX}
\mbox{\boldmath $p$}_k\cdot \mbox{\boldmath $x$}_m>N_o(r_o+\varepsilon_1).
\end{equation}

Since for $k\ge p$, $\mbox{\boldmath $p$}_k$ has no more than $N_o(r_o+\varepsilon_1)$ number of $1$'s, which implies $\mbox{\boldmath $p$}_k\cdot \mbox{\boldmath $x$}_m\le N_o(r_o+\varepsilon_1)$, the vectors that satisfy (\ref{PX}) must exist among $\mbox{\boldmath $p$}_k$ with $1\le k<p$. In words, for some $k$, $\mbox{\boldmath $q$}_k\cdot\mbox{\boldmath
$x$}_m>N_o(r_o+\varepsilon_1)$.
\end{proof}

Theorems \ref{Theorem1} and \ref{Theorem2} indicate that, if $\mbox{\boldmath $x$}_m$ is transmitted and $\mbox{\boldmath$\alpha$}\cdot\mbox{\boldmath $x$}_m>N_o\left (\frac{\varepsilon_2}{2}+(r_o+\varepsilon_1)(1-\frac{\varepsilon_2}{2})\right)$, for some $0 \le k < 1/\varepsilon_2$, errors-and-erasures decoding specified by $\mbox{\boldmath $q$}_k$ (where symbols with $\mbox{\boldmath $q$}_k(\alpha_i)=0$ are erased) will output $\mbox{\boldmath $x$}_m$. Since the total number of $\mbox{\boldmath $q$}_k$ vectors is upper bounded by a constant $1/\varepsilon_2$, the outer code carries out errors-and-erasures decoding only for a constant number of times. Consequently, a GMD decoding that carries out errors-and-erasures decoding for all $\mbox{\boldmath $q$}_k$'s and compares their decoding outputs can recover $\mbox{\boldmath $x$}_m$ with a complexity of $O(N_o)$. Since the inner code length $N_i$ is fixed, the overall complexity is $O(N)$.

The following theorem gives an error probability bound for one-level concatenated codes with the revised GMD decoder.

\begin{theorem}{\label{Theorem3}}
Assume inner codes achieve Gallager's error exponent given in (\ref{GallagerE}). Let the reliability vector $\mbox{\boldmath $\alpha$}$ be generated according to Forney's algorithm presented in \cite[Section 4.2]{ref Forney66}. Let $\mbox{\boldmath $x$}_m$ be the transmitted outer codeword. For large enough $N$, error probability of the one-level concatenated codes is upper bounded by
\begin{eqnarray}
P_e &\le & P\left\{\mbox{\boldmath$\alpha$}\cdot \mbox{\boldmath $x$}_m\le N_o\left(\frac{\varepsilon_2}{2}+(r_o+\varepsilon_1)\left(1-\frac{\varepsilon_2}{2} \right)\right)\right\} \nonumber \\
&\le & \exp \left[-N\left(E_c(R)-\varepsilon\right)\right],
\end{eqnarray}
where $E_c(R)$ is Forney's error exponent given by (\ref{EcR}) and $\varepsilon$ is a function of $\varepsilon_1$ and $\varepsilon_2$ with $\varepsilon\rightarrow0$ if $\varepsilon_1, \varepsilon_2\rightarrow0$.
\end{theorem}

The proof of Theorem \ref{Theorem3} can be obtained by first replacing Theorem 3.2 in \cite{ref Forney66} with Theorem \ref{Theorem2}, and then following Forney's analysis presented in \cite[Section 4.2]{ref Forney66}.

The difference between Forney's and the revised GMD decoding schemes lies in the definition of errors-and-erasures decodable vectors $\mbox{\boldmath $q$}_k$, the number of which determines the decoding complexity. Forney's GMD decoding needs to carry out errors-and-erasures decoding for a number of times linear in $N_o$, whereas ours for a constant number of times. Although the idea behind the revised GMD decoding is similar to Justesen's GMD algorithm \cite{ref Justesen72}, Justesen's work has focused on error-correction codes where inner codes forward Hamming distance information (in the form of an $\mbox {\boldmath $\alpha$}$ vector) to the outer code.

Applying the revised GMD algorithm to multi-level concatenated codes \cite{ref Blokh82}\cite{ref Barg05} is quite straightforward. Achievable error exponent of an $m$-level concatenated codes is given in the following Theorem.

\begin{theorem}{\label{Theorem4}}
For a discrete-time memoryless channel with capacity $\mathcal{C}$, for any $\varepsilon>0$ and any integer $m>0$, one can construct a sequence of $m$-level concatenated codes whose encoding/decoding complexity is linear in $N$, and whose error probability is bounded by
\begin{eqnarray}\label{MExponent}
\lim_{N\rightarrow\infty}-\frac{\log P_e}{N}\ge E^{(m)}(R)-\varepsilon, \qquad \qquad \qquad \qquad \qquad \nonumber\\
E^{(m)}(R)=\max_{p_X,r_o\in {\left [\frac{R}{\mathcal{C}},1\right]}}\frac{\frac{R}{r_o}-R}{\frac{R}{r_om}\sum_{i=1}^m\left [E_L\left((\frac {i}{m})\frac{R}{r_o},p_X\right )\right ]^{-1}} \nonumber \\
\end{eqnarray}
\end{theorem}

The proof of Theorem \ref{Theorem4} can be obtained by combining Theorem \ref{Theorem3} and the derivation of $E^{(m)}(R)$ in \cite{ref Blokh82}\cite{ref Barg05}.

Note that $\lim_{m\rightarrow \infty}E^{(m)}(R)=E^{(\infty)}(R)$, where $E^{(\infty)}(R)$ is the Blokh-Zyablov error exponent given in (\ref{BZBound}). Theorem \ref{Theorem4} implies that, for discrete-time memoryless channels, Blokh-Zyablov error exponent can be arbitrarily approached with linear encoding/decoding complexity.

\section{Conclusions}
We proposed a revised GMD decoding algorithm for concatenated codes over general discrete-time memoryless channels. By combining the GMD algorithm with Guruswami and Indyk's error correction codes, we showed that Forney's and Blokh-Zyablov error exponents can be arbitrarily approached by one-level and multi-level concatenated coding schemes, respectively, with linear encoding/decoding complexity.

\section*{Acknowledgment}
The authors would like to thank Professor Alexander Barg for his help on multi-level concatenated codes.



%




\end{document}